\documentclass[conference]{IEEEtran}
\IEEEoverridecommandlockouts

\usepackage[ruled,linesnumbered]{algorithm2e}

\def\BibTeX{{\rm B\kern-.05em{\sc i\kern-.025em b}\kern-.08em
    T\kern-.1667em\lower.7ex\hbox{E}\kern-.125emX}}
    
\usepackage{url}    
\usepackage{CJKutf8}  

\usepackage{xcolor}
\usepackage{pgfplots}
\usepackage{tikz}
\usepackage{comment}

\usepackage{filecontents}

\usepackage{multirow}
\usepackage{cite}
\usepackage{amsmath,amssymb,amsfonts}

\usepackage{textcomp}
\usepackage{xcolor}

\usepackage{subfigure}
\usepackage{pgfplots}
\usepackage{filecontents}

\usepackage{graphicx,subfigure}

\definecolor{bblue}{HTML}{4F81BD}
\definecolor{rred}{HTML}{C0504D}
\definecolor{ggreen}{HTML}{9BBB59}
\definecolor{ppurple}{HTML}{9F4C7C}

\newtheorem{theorem}{Theorem}
\newtheorem{definition}{Definition}
\newtheorem{lemma}{Lemma}
\newtheorem{proof}{Proof}

\usepackage{listings}

\usepackage{blindtext}
\usepackage{etoolbox,xstring,mfirstuc,textcase}
\usepackage{booktabs}
\usepackage{pgfplots}

\usepackage{tabularx,booktabs,makecell,multirow, caption}
\usepackage{enumitem} 

\newcolumntype{L}{>{\arraybackslash}X}

\usepackage{xcolor}
\usepackage{tikz}
\usepackage{pgfplots}

\definecolor{findOptimalPartition}{HTML}{D7191C}
\definecolor{storeClusterComponent}{HTML}{FDAE61}
\definecolor{dbscan}{HTML}{ABDDA4}
\definecolor{constructCluster}{HTML}{2B83BA}

\begin{document}

\title{\Large A Weak Consensus Algorithm and Its Application to \\ High-Performance Blockchain}


\author{\IEEEauthorblockN{Qin Wang \IEEEauthorrefmark{2}\IEEEauthorrefmark{4}$^\eth$\thanks{$\eth$: These authors contributed equally to the work.}, Rujia Li \IEEEauthorrefmark{1}\IEEEauthorrefmark{3}$^\eth$
}
\IEEEauthorrefmark{1}  \textit{Southern University of  Science and Technology}, Shenzhen, 518055, China.\\
\IEEEauthorrefmark{2} \textit{Swinburne University of Technology}, Melbourne, VIC 3122, Australia. \\
\IEEEauthorrefmark{3} \textit{University of Birmingham}, Edgbaston, B15 2TT, United Kingdom.\\
\IEEEauthorrefmark{4} \textit{HPB Foundation}, DUO Tower, 189352, Singapore. 
}

\maketitle
%

\begin{abstract}
A large number of consensus algorithms have been proposed. However, the requirement of strict consistency limits their wide adoption, especially in high-performance required systems. In this paper, we propose a weak consensus algorithm that only maintains the consistency of relative positions between the messages. We apply this consensus algorithm to construct a high-performance blockchain system, called \textit{Sphinx}. We implement the system with 32k+ lines of code including all components like consensus/P2P/ledger/etc. The evaluations show that Sphinx can reach a peak throughput of 43k TPS (with 8 full nodes), which is significantly faster than current blockchain systems such as Ethereum given the same experimental environment. To the best of our knowledge, we present the first weak consensus algorithm with a fully implemented blockchain system.

\end{abstract}
\smallskip

\section{Introduction}
The consensus mechanism is a critical component in distributed systems, providing a powerful means of establishing agreement as to the network’s current state. With the promotion of blockchain, consensus mechanisms obtain tremendous attention due to their influential roles in secure token transferring. Generally, two mainstream types of the consensus algorithms are identified \cite{bano2019sok}\cite{vukolic2017rethinking}, namely, the classic Byzantine Fault Tolerant (BFT) protocols  \cite{CastroL99}\cite{castro2002practical} and the newly proposed Nakamoto consensus (NC) \cite{garay2015bitcoin} such as PoW~\cite{bonneau2015sok,nakamoto2008peer}, PoS~\cite{kiayias2017ouroboros}, PoA~\cite{toyoda2020function}, \textit{etc}. However, blockchain systems adopting these algorithms suffer from low-performance issues due to massive communication or intensive computation. For example, Bitcoin requires competitive computations to decide the valid chain, whose rate is limited to 7 transactions per second (TPS) \cite{nakamoto2008peer}. The limitations 
greatly hurdle the widespread adoptions in real scenarios. This leads us back to their core mechanisms.

BFT protocols have been proposed to achieve consensus in the presence of malicious nodes, where the tolerance is maximal $2/3$ of the total nodes. BFT protocols require the negotiation process for final decisions. A typical system, PBFT \cite{castro2002practical}, is illustrated in \underline{Fig.\ref{fig:bft}.a}. The leader sends a proposal to replicas, and replicas distribute their replies. Then, after receiving valid replies over the predefined threshold, a replica broadcasts his status (whether ready for the new state) to others. The decision is made once upon the received \textit{commit} messages exceeds the threshold. Time consumption in such a process is unpredictably unstable due to factors like network delay. Interactive communications consequently limit the performance of the BFT consensus and increase the communication overhead.

Nakamoto consensus, in recent decades, stands out in one critical aspect thanks to its remarkable simplicity. NC breaks the assumption that only the closed committee can conduct the consensus, instead, it enables all participants to get involved in the consensus process. NC protocols patterned after Bitcoin \cite{nakamoto2008peer} do not reach consensus with finality. NC protocols remove the interactive model and adopt a competition rule -- \textit{the longest chain wins}. As shown in \underline{Fig.\ref{fig:bft}.b}, blocks generated by miners are randomly attached to their ancestors. Only the chain who has most descendants survives, whereas other competitive sub-chains are abandoned. The finality is progressively achieved by letting blocks bury deep enough. Thus, conflict solving in NC significantly slows down the confirmation of blocks.

We observe that protocols based on these two types of consensus mechanisms follow the same principle that: only one block is deemed as confirmed at one round (equal block height in NC). This greatly constrains their overall performance, since the procedure of conflict solving and total ordering serving for \textit{strong consistency} costs much more time than expected. Such mechanisms hurdle their widespread adoptions~\cite{kogias2016enhancing}\cite{gervais2016security}, particularly, for some high performance required scenarios~\cite{wenhao2020blockchain}. To mitigate such a limitation, we ask the following question,

\smallskip
\textit{Is it possible to propose a consensus algorithm to improve the performance by weakening the guarantee of consistency?}
\smallskip

Intuitively, the answer should be ``No''. The state consistency is the core property for the consensus mechanism. Strict consistency ensures the distributed network reaches an agreement on the total order of transactions in the presence of fault maintainers and adversarial network delay. All distributed nodes have the same global view at each specific height. This guarantees that states are transited in an organized and managed way, supporting upper-layer establishments like smart contracts. Disordered transactions, on the contrary, indicate ambiguous states where users may feel confused when invoking the blockchain service. For example, Alice sends a transaction to Bob. If this transaction is stored in more than one block, Bob cannot know which position provides a valid transaction. However, in some scenarios, strict consistency is not firmly required, such as the blockchain-based certificate system. The main target of the certificate prover is to confirm that a certificate is indeed stored in the chain. The specific position of this certificate does not matter; even a duplicated storage of certificates is allowed. It should be noted that the partial consistency in several DAG-structure projects \cite{bagaria2019prism}\cite{yu2018ohie} is still sensitive for the position of transactions since their upper-layer applications \cite{wang2020sok}, such as token transferring, are still based on a fixed sequence of transactions.

\begin{figure}[htb!]
    \centering
    \includegraphics[width=0.83\linewidth]{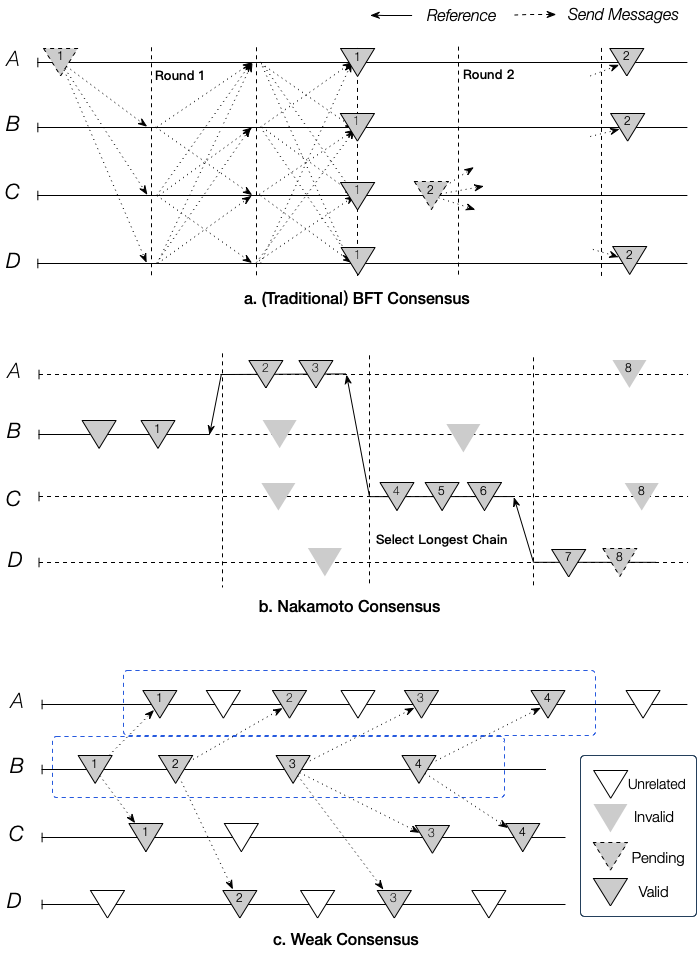}
   \caption{Consensus Mechanisms}
    \label{fig:bft}
\end{figure}

In this work, we propose a new consensus mechanism, called \textit{weak consensus}, to fit for the aforementioned scenarios. Our design weakens the guarantee of strict consistency and relaxes the property of persistence \cite{garay2015bitcoin}. Weak consensus mainly guarantees that the relative sequences of blocks in one individual chain remain consistent with that in the other chains. As illustrated in \underline{Fig.\ref{fig:bft}.c}, node $B$ creates a serial of blocks $1, 2, 3, 4$. Our goal is to ensure that the sequence of  $(B1 \to B2 \to B3 \to B4)$ can be correctly maintained across chains, no matter how many blocks (generated by other nodes) are inserted between them. Blocks in our model are required to receive replied messages from peers, saying that they have successfully stored the blocks. Whenever a block collects \textit{commit} messages more than the threshold, it is deemed as confirmed. To demonstrate the robustness of our consensus, we formally define the properties \textit{relative persistence} and \textit{liveness}, inspired by \cite{garay2015bitcoin}. Relative persistence focuses on the relationship between a predecessor and its successor, ensuring the correctness of the relative position. Liveness guarantees that all nodes would eventually agree on the relationship of blocks. Furthermore, we apply this algorithm to a blockchain system, called \textit{Sphinx}, with a full implementation. In summary, we make the following contributions.

\begin{itemize}
\item[$\diamond$] We identify the reasons causing the performance bottleneck of current consensus protocols, and propose a weak consensus mechanism to enable parallel processing of chains by weakening the guarantee of strict consistency.

\item[$\diamond$] We apply our designed weak consensus algorithm to a high-performance blockchain system with the properly defined model and  strictly proved robustness. 

\item[$\diamond$] We provide a full implementation\footnote{Github \url{https://github.com/rjgeek/sphinx}} of Sphinx with evaluations on its performance. The results show that our system is practically efficient with 43k TPS (with 8 full nodes). 

\end{itemize}

The rest of our paper is structured as follows: A high-level algorithm design is illustrated in Section~\ref{sec-design}. The system protocol and its implementation are shown in Section~\ref{Sec-archi} and \ref{sec-imple}, respectively. The security analysis is given in Section ~\ref{sec-analysis}, and evaluations of our system are provided in Section \ref{sec-evaluation}, followed by case examples in Section \ref{sec-example}. The related studies are discussed in Section \ref{sec-relatedwork}. Finally, the conclusion and future work are discussed in Section ~\ref{sec-conclusion}.

\section{Weak Consensus Algorithm}
\label{sec-design}

This section provides security assumptions and the general construction of our consensus algorithm with corresponding security properties.
\subsection{Notations}
We denote the nodes in our protocol as $\mathit{N}$ and identify each of them as $\{\mathit{N}_0,\mathit{N}_1,\dots, \mathit{N}_n\}$, where $n$ is the index of committee members satisfying $n=3f+1$. Let $i$ be a growing integer satisfying $0 \leq i \leq r$ where $r$ is the index of states and $j$ be an integer satisfying $0<j\leq n$. Assume that $\mathbb{M}$ is the message space, $\mathbb{S}$ is the state space and $\mathbb{R}$ is the reference space. $M$, where $M\in\mathbb{M}$, is the message proposed by some node. $PF$ is the proof of a successful insertion of $M$. $S$ represents a confirmed state satisfying $S\in \mathbb{S}$. $S_{N_j}^i$ is the $i$-th confirmed state in the node $N_j$, where $S_{N_j}^i \in \{ \mathbb{S} | S_{N_1}^0,\dots,S_{N_1}^r; \dots; S_{N_n}^1,\dots,S_{N_n}^r \}$. These two parameters are used to locate a specific state in the network. $S^{r}_{\{{N_0},\dots,{N_n}\}}$ refers to the states received from other nodes in current round $r$. $\Downarrow$ is the reference which indicates the relative positions between two states. Specifically, $\Downarrow^B_A$ is the reference pointing from $B$ to $A$. It defines a \textit{happens-before} relationship, such that $A$ happens before $B$. More specifically, $\Downarrow^{S^{y}_{N_j}}_{S^{x}_{N_j}} $ means that $S^x_{N_j}$ is an ancestor of $S^y_{N_j}$ where $0<x<y\leq r$. Further, $\Downarrow^{S_{N_j}^{i}}_{S_{\star}^{i-1}}$ represents a set of references including the edges from the state $S_{N_j}^{i}$ to the states $S_{N_0}^{i-1},\dots,S_{N_n}^{i-1}$ (\textit{a.k.a.}, the out-degree edges of $S_{N_j}^{i}$). Correspondingly,  $\Downarrow_{S_{N_j}^{i-1}}^{S_{\star}^{i}}$ contains all the edges from the states $S_{N_0}^{i},\dots,S_{N_n}^{i}$ to the state $S_{N_j}^{i-1}$ (\textit{a.k.a.}, the in-degree edges of $S_{N_j}^{i-1}$). 

\subsection{Security Assumption}

We assume that the honest nodes will always conduct honest behaviors, where the messages sent to the peers are correct. As for the underlying network, we follow the implicit assumption of a partial synchronous network. In particular, the network of honest nodes in our system are well connected, and the communication channels between honest nodes are unobstructed. Messages from honest broadcasters may be delayed, but they will eventually arrive at others within known maximum delay $\delta$ \cite{pass2017analysis}. Our algorithm follows the basic design of classic BFT-style protocols with the aim to tolerate one-third of Byzantine nodes. Specifically, we assume that there are $3f+1$ nodes in total, and the number of participated nodes is fixed. It indicates that the dynamics of peer participation, or churn, are out of our consideration. Also, we assume that at least $2f$ of them work honestly, where $f$ is the number of Byzantine nodes.

\subsection{Protocol Overview}

\begin{table*}[htb!]
  \centering\setlist[itemize, 1]{noitemsep,topsep=0pt, wide=0pt, leftmargin =\dimexpr \labelwidth+ 2\labelsep\relax, after=\vspace*{\dimexpr-4\partopsep}}
  \caption{\small{Comparison between PBFT algorithm and Our algorithm}}
  \begin{tabularx}{\textwidth}{lLL}
    \toprule
      & \thead{PBFT algorithm~\cite{CastroL99}} & \thead{Our algorithm} \\
    \midrule
    \textbf{Request Stage}
    &
    \begin{itemize}
    \item A leader is required. A client sends a request to the primary node. If the primary node has changed/rotated, it will broadcast the request message to all replicas.   
    \end{itemize}
    &
    \begin{itemize}
    \item No leader exists in our algorithm. Every node acts similar behaviors. A client sends a request to a random node, where a request message is represented as the relative position.
    \end{itemize} \\
    
    \midrule 
    \textbf{Normal Case}
    &
    \begin{itemize}
    \item \textit{Pre-Prepare:} The primary node puts the pending requests in a total order and initiates agreement by sending \textit{Pre-prepare} message to all replicas.
    \item \textit{Prepare:} Replica acknowledges the receipt of a \textit{Pre-prepare} message by
sending \textit{Prepare} message to other replicas.
    \item \textit{Commit:} Replica acknowledges the reception of $2f$ \textit{Prepare} message matching a valid pre-prepare by broadcasting the \textit{Commit} message to peers.
    \end{itemize}
    &
    \textbf{Every node} executes the following actions in parallel.
    \begin{itemize}
    \item \textit{Pre-Prepare:} The same algorithm with PBFT.
    \item \textit{Prepare:} The same algorithm with PBFT.
    \item \textit{Commit:} The same algorithm with PBFT.
    \end{itemize} \\

    \midrule 
    \textbf{View Change}
    &
    \begin{itemize}
    \item It ensures that the system can always proceed by allowing replicas to change the leader so as to not wait indefinitely for a faulty primary.
    \end{itemize}
    &
    \begin{itemize}
    \item No needs for the view change progress. Alternatively, the complementary mechanism was adapted to ensure the correct relative persistence to be held.
    \end{itemize} \\
    
    \midrule
    \textbf{Garbage collection}
    &
    \begin{itemize}
    \item The checkpoint mechanism is used to ensure the safety condition to be held.
    \end{itemize}
    &
    \begin{itemize}
    \item The timeout mechanism is used to ensure the liveness condition to be held. 
    \end{itemize} \\
    
    \bottomrule 
  \end{tabularx}
  \label{tab:GPS_INS_comparison}
\end{table*}

The protocol is modeled as a state machine which is replicated across distributed nodes. Each node in the network maintains a message log containing the \textit{accepted} message and the current state. Meanwhile, in our algorithm, a node must maintain the states received from other nodes. We present our protocol by following the description of PBFT~\cite{CastroL99}\cite{castro2002practical}\cite{sukhwani2017performance}: the protocol proceeds in rounds, and each round has three phases, namely \textit{Pre-prepare}, \textit{Prepare} and \textit{Commit}. We provide a protocol overview of our protocol as follows. 

\begin{itemize}
\item[-] \textbf{Pre-prepare.} The primary node receives the client requests and inserts the such messages into the local chain. Then, the node creates a \textit{Pre-prepare} message to claim the relative position between two client messages. Subsequently, it broadcasts the signed message to peers.

\item[-] \textbf{Prepare.} The node receives the $\textit{Pre-prepare}$ message and checks the integrity, correctness, and validity. When the received $\textit{Pre-prepare}$ message passes the verification, the node updates his local-stored state and broadcasts the replied $\textit{Prepare}$ message to claim the correct relative position. Otherwise, the node aborts it. 

\item[-] \textbf{Commit.} If any node receives a quorum $2f + 1$ of valid $\textit{Prepare}$ messages from peers (possibly including his own) within specified time interval, this node confirms the proposed decision by broadcasting a replied $\textit{Commit}$ message. When a node collects more than $2f + 1$ \textit{Commit} messages, this node transits the state and replies to clients with updated states.

\end{itemize}

\smallskip
\noindent\textbf{Complementary Mechanism.} Our protocol aims to achieve an eventual confirmation of the relative relationship between two states. It is possible that the relative relationship cannot be committed due to the lack of $2f + 1$ matched \textit{Commit} messages. \underline{Fig.\ref{fig:conflict}} provides an example to explain the flaws. We assume that there are three nodes, and the relative position of states $S_{N_2}^{1}$ and $S_{N_2}^{2}$, saying $\Downarrow^{S^{2}_{N_2}}_{S^{1}_{N_2}}$, in the node $N_2$ have already been confirmed by the node $N_3$. By the design of previous protocols, Sphinx achieves an agreement by receiving $2f + 1$ replies. However, the node $N_1$ may store the conflicting states, namely, $\Downarrow^{S^{1}_{N_1}}_{S^{2}_{N_1}}$. Note that, here, we assume $\Downarrow^{S^{1}_{N_2}}_{S^{2}_{N_2}}$= $\Downarrow^{S^{1}_{N_1}}_{S^{2}_{N_1}}$, indicating that the states sent from $N_2$ are finally confirmed in the node $N_1$, and accordingly subscripts are changed. Without re-pulling the state from other nodes when the bound set in the counter is exceeded, the conflicting states will never be able to be reversed. Thus, the complementary mechanism reattaches the state $S_{N_1}^{2}$, and the updated $S_{N_1}^{2}$ will replace the outdated version. The relative positions in $N_1$ are thus returned to correct positions, saying $\Downarrow^{S^{2}_{N_1}}_{S^{2}_{N_1}}$ = $\Downarrow^{S^{2}_{N_2}}_{S^{1}_{N_2}}$ = $\Downarrow^{S^{2}_{N_3}}_{S^{1}_{N_3}}$. The above procedure is compulsory for $N_1$. If $N_1$ cannot accept relative positions from honest nodes $N_2$ and $N_3$, the new messages that follow $N_1$ will not be accepted by $N_2$ and $N_3$ since his \textit{Pre-prepare} message is based on the latest state from others. More details are provided in our security analysis.

\begin{figure}[htb!]
    \centering
    \includegraphics[width=0.8\linewidth]{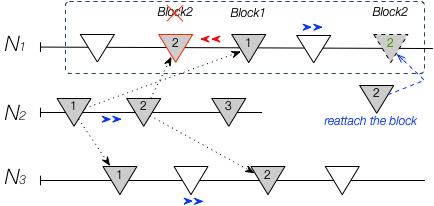}
    \caption{Complementary Mechanism}
    \label{fig:conflict}
\end{figure}

\smallskip
\noindent\textbf{Highlighted Differences.} Our protocol differs from PBFT in four aspects (see \underline{Table.\ref{tab:GPS_INS_comparison}}): (a) Our protocol is an asynchronously leaderless Byzantine agreement protocol. Instead of relying on a single leader, Sphinx removes the leader-associated phases enabling every participant involved in the consensus procedures. Thus, every participant does not need to wait for the latest state synchronized from others. (b) Every consensus node in the network conducts the similar behaviors (\textit{pre-prepare}, \textit{prepare}, \textit{commit}). These nodes independently proceed but mutually interact with each other by cross-references. (c) We remove the auxiliary mechanisms (such as \textit{checkpoint}, \textit{view-change}) of PBFT protocols. Instead, we provide a brief complementary mechanism to solve conflicts like reserved positions between two states. (d) Our protocol weakens the assumption of \textit{strong consistency}, with the gain of higher performance and lower confirmation time. We follow the \textit{liveness} property from PBFT. Further, we introduce a new security definition \textit{relative persistence}, which is inspired by the properties of $\textit{agreement}$~\cite{CastroL99} and $\textit{persistence}$~\cite{garay2015bitcoin}.

\subsection{Security Properties}
Our algorithm weakens the strong guarantee of consistency by reaching a partial consistency instead of a linear consistency. The procedure of total ordering is no longer needed in our consensus. Each node individually creates new states and simultaneously stores the remote (other's) states. We allow one state to be stored in multiple positions across parallel chains. The key idea is to keep the consistency of the relative position between two states. Based on that, we formalize our algorithms by two properties: \textit{relative persistence} and \textit{liveness}. \textit{Relative persistence} means the relative positions of two states are irreversible once enough honest nodes report it as confirmed. \textit{Liveness} means once the relative position between two states is confirmed by one honest node, it should be eventually confirmed by other honest nodes in the network.

\smallskip
\begin{definition}[Relative persistence] 
\textit{Sphinx achieves the property of relative persistence, if for all relationship in $\mathbb{R}$, there exists a negligible function $negl(\lambda)$ such that $adv_{\mathbb{N}}^{{\mathbb{R}}}(\lambda) < negl(\lambda)$, where $adv_{\mathbb{N}}^{\mathbb{R}}(\lambda)$ is the advantage in which the decisions on the same relationship $\Downarrow^{S^y_{N_i}}_{S_{N_i}^{x}}$ made by any two honest nodes are conflicting. Here, $0<x<y\leq r$. }
\end{definition} 

The \textit{relative persistence} property ensures that as soon as the relative position between two states has been confirmed by honest node, this relationship will ultimately be confirmed in every node in the network with high probability. The property guarantees that the relative positions remain consistent between two states across paralleled chains.

\smallskip
\begin{definition}[Liveness] 
\textit{Sphinx achieves the property of liveness, if for all PPT honest nodes, there exists a negligible function $negl(\lambda)$ such that $adv_{\mathbb{N}}^{\mathbb{R}}(\lambda) < negl(\lambda)$, where $adv_{\mathbb{N}}^{\mathbb{R}}(\lambda)$ is the advantage that the honest node does not accept the correct relationship $\Downarrow^{S^y_{N_i}}_{S_{N_i}^{x}}$. Here, $0<x<y\leq r$.}  
\end{definition}

\textit{Liveness} guarantees that all nodes eventually agree on a unique relationship regarding each chain. The unique relationship represents the relative position between a state and its ancestor. The term \textit{eventually} indicates that it may take a sufficient amount of time (within the upper bound of $\delta$) to reach the agreement. The property ensures that a state will either be abandoned or accepted, instead of permanently pending status.

\section{Sphinx System}
\label{Sec-archi}
In this section, we firstly introduce the cryptographic building blocks used to build the scheme. Then, we present a high-performance blockchain system that adopts the weak consensus algorithm.

\subsection{Cryptographic Building Blocks}

\noindent \textbf{Merkle Hash Tree.} Merkle tree \cite{ryan2014enhanced} is a tree based on the one-way cryptographic hash function. The data is compressed in the form of the cryptographic hash. Every node in the Merkle tree is labeled with a hash composed of its children nodes. This technique enables efficient and secure verification for the transactions in the blockchain system.

\smallskip
\noindent\textbf{Signature Scheme.} A signature scheme consists of three probabilistic polynomial-time algorithms (\textit{kgen}, \textit{sig}, \textit{verify}), where \textit{kgen} generates a private signing key and a public verification key, \textit{sig} outputs a signature on the message, and \textit{verify} is used to check the signature. 

\subsection{Entities.}  
Sphinx mainly consists of two types of nodes, namely \textit{blockchain node} and  \textit{client node}. The client node is the creator of the message, and it allows to send a request to the ledger recorder and wait for the replies after the consensus is completed. The blockchain node is responsible for two functionalities: \textit{record} and \textit{validate}. The former functionality is used to record the local chain. Another one is employed to check the correctness of the recording progress.

\subsection{System Design}
Our concrete construction, \textit{Sphinx}, is based on the weak consensus algorithm.  The state in Sphinx is instantiated as the block, which is validated and confirmed by peer nodes. The message inherits the classic blockchain structure, which includes the fields of address, timestamp, metadata \textit{etc}. Every chain has two types of references in our system: the one pointing to its own parent block and the other pointing to the peers. Explicitly, our system embraces cross-referencing to increase the blockchain's security, where multiple nodes simultaneously generate their own chains in parallel and validate the blocks of other nodes. The cross-reference ensures that each chain can mutually validate others' behaviors, such as whether they maintain a consistent sequence of blocks by checking their previous Merkle roots. This mechanism guarantees the consistency of relative positions between two blocks. The concrete protocol is presented as follows.

\smallskip
\noindent\textbf{Pre-prepare.} Each node maintains an individual message pool and an independent ledger. When a node $N_p$ receives a request (message) $M$ from the client, it firstly checks the message's syntax to ensure the correct execution. If passed, it sorts the received client messages in the local pool. These messages are ordered by their timestamps (\textit{e.g.,} Lamport timestamps~\cite{CastroL99}), in which the latest messages have higher priorities than the earlier ones. The first received message is processed in the algorithm, whereas the conflicting/duplicated messages are discarded. Based on the ordered sequence, the node assigns a valid sequence number $SN$ to $M$. We emphasize that the sequence number $SN$ is a local variable used to represent the index $r$ of a single chain. It is different from the global variable $SN$ defined in BFT algorithm~\cite{CastroL99}\cite{castro2002practical}. The index $r$ in our system helps to locate a block in the local chain, while locating a block in the global view requires an additional parameter of chain id $N_j$ to form a coordinate $(r,N_j)$. 

After that, it inserts the message $M$ to its local chain. The insertion mainly merges the hash of $M$ and the states of remote blocks $(S_{N_0}^r, S_{N_1}^{r}, \dots, S_{N_n}^{r})$ into a new \textit{Merkle Tree}. Next, it signs the Merkle root and appends this block to a public log by generating the proof $PF$ which contains three types of proofs: (a) $P_M$ used to prove the log contains $M$; (b) $P_B$ as an extension of old blocks; (3) $P_{S}$ to prove $S$ exists in trees. 

The message $M$ and the states $(S_{N_0}^r, S_{N_1}^{r}, \dots, S_{N_n}^{r})$ are stored in the leaves of the Merkle tree from left to right in chronological order. Thus, the proof can be easily calculated. A Merkle tree contains the items of $M$ and the states $(S_{N_0}^r, S_{N_1}^{r}, \dots, S_{N_n}^{r})$. In our design, these items are solely stored at leaves. After that, the current node broadcasts the \textit{Pre-prepare} message $\langle \mathsf{Pre}-\mathsf{prepare}, M, SN, S_{N_p}^{r}, PF, \Downarrow^{S_{\star}^{r}}_{S_{N_p}^{r-1}}\rangle$ to the peer nodes, where $\Downarrow^{S_{\star}^{r}}_{S_{N_p}^{r-1}}$ represents the relative position between the $r$-th state and $(r-1)$-th state.

\smallskip
\noindent\textbf{Prepare.} Assume that a random node $N_q$ receives the $\textit{Pre-prepare}$ message $\langle \mathsf{Prepare}, S_{N_p}^{r}, M, PF, \Downarrow^{S_{\star}^{r}}_{S_{N_p}^{r-1}} \rangle$ from the node $N_p$. The node $N_q$ validates the correctness of $\textit{Pre-prepare}$. The algorithm checks whether: (a) the signature of the $S^{\star}$ is correctness; (b) the message $M$ has been inserted to $N_p$; (c) the previous state $S_{q}^{\star}$ has been inserted to $N_p$; (d) the state of the claimed message has no conflict. When the received $\textit{Pre-prepare}$ message is checked successfully, the node $N_q$ updates his local-stored state $S_p$, and inserts the received message to his local log. Then, it generates and broadcasts a reply $\textit{Prepare}$ message  $\langle\mathsf{Prepare}, SN, d(SN)\rangle$, in which $d(SN)$ means the hash digest of the state. Otherwise, the node aborts the $\textit{Pre-prepare}$ message. Note that the same $\textit{Pre-prepare}$ message can only be accepted once, and the duplicated $\textit{Pre-prepare}$ message will be discarded. 

\smallskip
\noindent\textbf{Commit.} If a node receives a quorum $2f + 1$ of valid $\textit{Prepare}$ message from different nodes (possibly including his own), it confirms the proposed message and broadcasts a $\textit{Commit}$ message $\langle\mathsf{Commit},SN,d(SN)\rangle$. Then, this node collects the $\textit{Commit}$ messages from all peers. Once exceeding the threshold ($2f + 1$), the node accepts the updated state. Then, this node replies to the client with new states. 

\smallskip
\noindent\textbf{Complementary Mechanism} The aforementioned phases are insufficient since some malicious $2f-1$ nodes may store the wrong (opposite) relative positions, making our system fail in satisfying the security property of relative persistence. It is possible that the relative relationship cannot be committed due to the lack of $2f + 1$ matched \textit{Commit} messages. The states, as a result, will be never terminated. To solve this problem, we introduce our complementary methods. On the one side, each message is embedded with a counter. If the message fails due to the lack of enough confirmation, the procedure of rebroadcast will be launched, and the counter increases each time of a retry. If the accumulated value is greater than the bound set in the counter, the node will pull the newest state from other peers and accept the reversed relationship and rebroadcast it. If a node collects more than $2f + 1$ \textit{Commit} messages on the reversed position, the node replies to clients with updated states. Otherwise, the message will be aborted and send a \textit{Failure} message to the client. On the other side, when the waiting time exceeds the predefined time bound in the counter, the message is aborted with a \textit{timeout} message sent to the client. The complementary mechanism is  essential to achieve the properties of \textit{relative  persistence} and \textit{liveness}.

\section{Implementation}
\label{sec-imple}
To evaluate our Sphinx, we have implemented the system in Go language with 32,000+ lines of code. We have developed full functionalities of a classic blockchain system, including account configuration, consensus mechanism, peer to peer network, user interface, \textit{etc.}  We employ Go’s built-in hash function \textit{SHA-256} and elliptic curve digital signature algorithm \textit{secp256k1}. Here, we focus on key functions to present a skeleton of our implementation. Example code segments together with the workflow are illustrated in
\underline{Fig.\ref{fig:imple}}. To be specific, $\mathsf{ValidateState}$ validates the changed state after the state transition, such as the receipt roots and state roots. The function will return a database handle if the validation turns out a success. Otherwise, an error is returned. $\mathsf{ValidateBody}$ validates the uncle blocks and verifies their header's receipts. The headers are assumed to be already validated at this point. $\mathsf{NewBlockChain}$ returns a fully initialized blockchain by loading the information in the database. It initializes the default validators. $\mathsf{FastSyncCommitHead}$ inserts the committed head block to others by the form of hashes. $\mathsf{GetBlocksFromHash}$ returns the block corresponding to hash and up to $n-1$ ancestors. $\mathsf{InsertHeaderChain}$ attempts to insert the headers of parallel chains into the local chain. $\mathsf{Insert}$ inserts or rejects a new header of the block into the current chain. $\mathsf{Confirm}$ aims to ensure whether threshold conditions are satisfied by a block.

\begin{figure*}[htb!]
   \centering
    \includegraphics[width=0.95\linewidth]{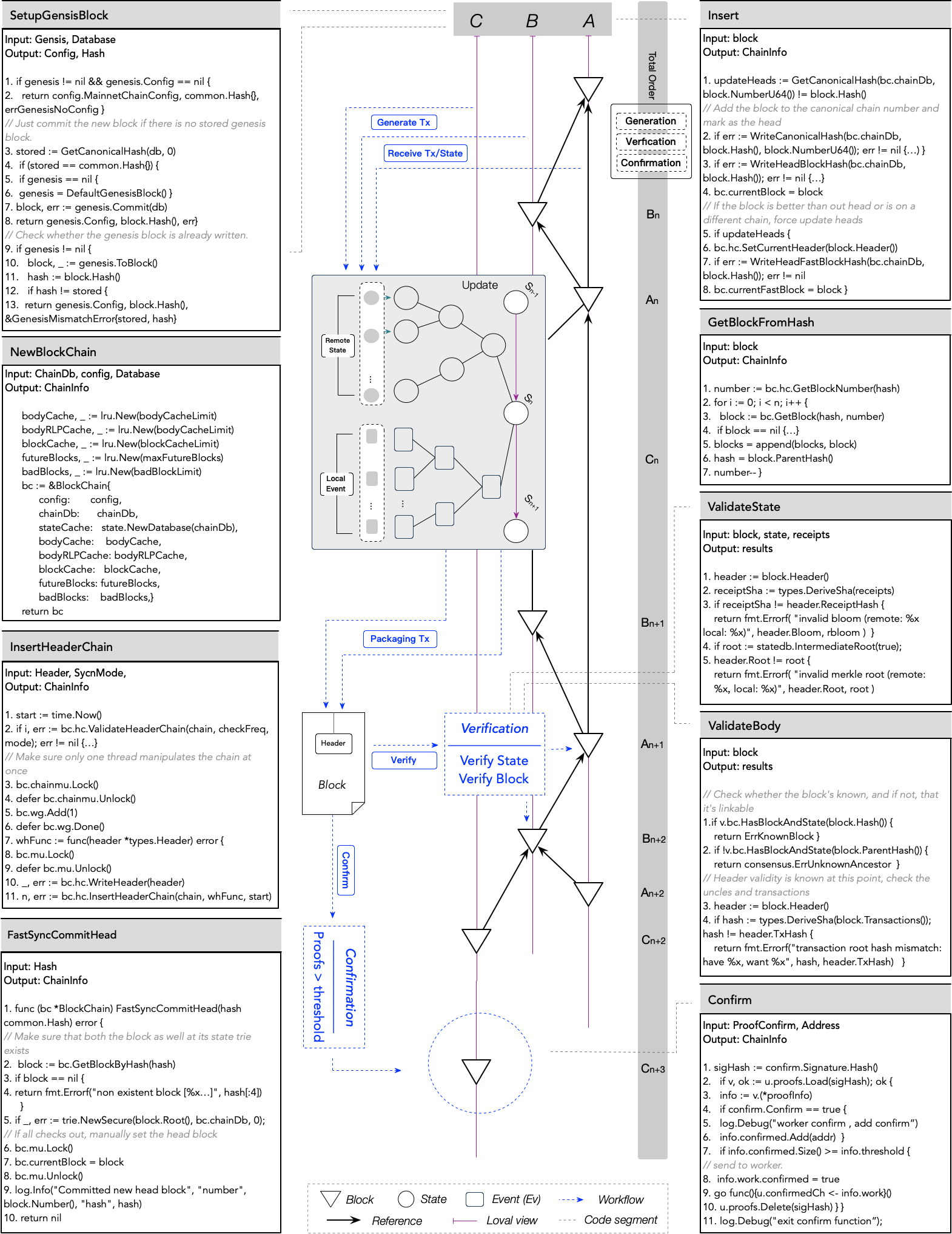}
    \caption{Sphinx Implementation}
    \label{fig:imple}
\end{figure*}

\section{Security analysis}
\label{sec-analysis}
In this section, we prove that our protocol satisfy~\textit{relative persistence} and~\textit{liveness}. We firstly assume that honest nodes are consistent for their commits, which means if an honest node accepts a relative state, all his commits in this iteration and the following iterations are consistent. Formally, we define the above intuition in the Lemma~\ref{lem}.

\smallskip
\begin{lemma} 
\label{lem}
\textit{Suppose a node $N_i$ is the first honest node to commit a relationship $\Downarrow^{S^{y}_{N_i}}_{S^{x}_{N_i}}$ for the relative positions between the $y$ and $x$. In all subsequent iterations, all commits from peers can construct valid decisions on the relationship   $\Downarrow^{S^{y}_{\star}}_{S^{x}_{\star}}$}.
\end{lemma}

\smallskip
\begin{theorem} \textbf{(Relative persistence)} \textit{If the relative position of two state $y$ and $x$ is accepted by the node $N_i$ in iteration $r$ and by the node $N_j$ in $r+1$, respectively, their decisions on the relationship are the same, represented as  $\Downarrow^{S^{y}_{i}}_{S^{x}_{i}} = \Downarrow^{S^{y}_{j}}_{S^{x}_{j}}$.}
\end{theorem}

\begin{proof} 
We prove the theorem by contradiction induction. We assume the relationship $\Downarrow^{S^{y}_{i}}_{S^{x}_{i}}$ is accepted by the node $N_i$. Similarly, we assume the $\Downarrow^{S^{y}_{j}}_{S^{x}_{j}}$ is accepted by the node $N_j$. We show that, in current iteration and all subsequent iterations, the reported relationships from $N_i$ and $N_j$ are consistent, namely, $\Downarrow^{S^{y}_{i}}_{S^{x}_{i}} = \Downarrow^{S^{y}_{j}}_{S^{x}_{j}}$, and further, no valid relationships other than the reported one can be agreed upon.

Suppose the property of relative persistence is violated, which indicates the relationship holds $\Downarrow^{S^{y}_{i}}_{S^{x}_{i}} \neq \Downarrow^{S^{y}_{j}}_{S^{x}_{j}}$. From the above assumption, we know the relationship $\Downarrow^{S^{y}_{i}}_{S^{x}_{i}}$ from the node $N_i$ has been accepted, which means $N_i$ received $2f + 1$ valid commit replies in current iteration. Among these replies, at least one of the commits comes from an honest node (assume $N_k$). Thus, $N_k$ must have received a \textit{right} relative position between two states, and forwarded this relationship to all other nodes. If $\Downarrow^{S^{y}_{k}}_{S^{x}_{k}} \neq \Downarrow^{S^{y}_{j}}_{S^{y}_{j}}$, other honest nodes can immediately detect the mismatch of relationships from different proposals. A wrong located state (with the same sequence number) is in a reversed position. The nodes then reattach the state until they obtain the consistent relationship. The dishonest nodes will never collect more that  $2f + 1$ valid commit messages unless the majority honest nodes become traitors in the current and following iterations. However, this situation contradicts \textit{Lemma 1}. Thus, we have $\Downarrow^{S^{y}_{i}}_{S^{x}_{i}} = \Downarrow^{S^{y}_{j}}_{S^{x}_{j}}$.
\end{proof} 

\medskip
Now, we move on to \textit{liveness} properties and explain how our algorithm guarantees all honest nodes agree on the same relative relationship and reach the termination.
\smallskip
\begin{theorem}
\textbf{(Liveness)} \textit{If a correct relationship $\Downarrow^{S^{y}_{i}}_{S^{x}_{i}}$ is committed, every honest node will eventually accept it.}
\end{theorem}

\begin{proof} 
Suppose that the property of \textit{liveness} is violated, which means only a small fraction of nodes (less than the threshold) or even none of the nodes, accept the final decision on the relationship  $\Downarrow^{S^{y}_{\star}}_{S^{x}_{\star}}$ between the states $y$ and $x$. Equally, a majority of nodes reject or have no responses for the commit decision after a sufficient period of time. 
We show that, in all subsequent rounds, the malformed relative relationship from the dishonest nodes will never be accepted. 

If a randomly selected group of dishonest nodes (more than $\lfloor \frac{3f+1}{2}\rfloor $) reject the decision of the relationship $\Downarrow^{S^{y}_{\star}}_{S^{x}_{\star}}$, an honest node $N_i$ will never be terminated. This is due to the fact that the state of $N_i$ has to be confirmed by enough nodes, which requires at least $2f + 1$ commit messages from the peers. The confirmed message is based on the other honest nodes' states. If the majority of nodes reject a correct proposal, any honest node cannot get confirmed, either. Thus, when the times out, it must reverse the relative relationship and restart the generation phase to achieve the final confirmation. 
\end{proof}

\section{Evaluation}
\label{sec-evaluation}

\smallskip
\noindent\textbf{Experimental Configurations.} Our experiments are conducted on 8 Dell R730 rack servers in a local cluster, with dual 2.1 GHz Opteron CPUs. The bandwidth is connected with 1 Gbps switched Ethernet. The operating system is running based on Ubuntu 16.04.1 LTS version.

\smallskip
\noindent\textbf{Performance Evaluation.}
The throughput represents the rate of transactions being confirmed at a certain time interval. We adopt the log-based approach~\cite{zheng2018detailed} and the concept of transactions per second (TPS) to measure the rate. In our experiments, we set the production of transactions at a constant rate and calculated the confirmation time of transactions in a fixed time. The time is measured in seconds via wall clock running time. To achieve a fair test result, we repeat the tests $300$ times. The results show that the average throughput of Sphinx reaches $43$k TPS with $8$ full nodes and drops to around $5000$ TPS when given $64$ full nodes. However, our system is greatly faster than Ethereum. Given the same testing environment, Ethereum (version 1.9.25, released on December 11, 2020) only reaches $355$, $311$, $268$, $91$ TPS under the setting of $8$, $16$, $32$, $64$ nodes, respectively. To explore the reason behind the high throughout, we further evaluate the performance of each individual algorithm.

Firstly, in the Pre-prepare stage, the transaction is added into a block. The evaluation results (see \underline{Fig.\ref{fig:bcr}}) show that is takes approximately $15$ milliseconds to finish the block generation algorithm, making our system reach an extremely high throughput. The high-speed generation rate is based on the parallel processing mechanism. Each node generates and maintains his own ledger, even if it does not consist of the latest state header. In contrast, the node in classic blockchain systems such as Ethereum cannot mine until it obtains the latest block header, which causes a severe delay.

Then, we consider the verification time and message broadcasting time in the Prepare stage. The verification time represents the length of time in verifying specific fields of states. In particular, it covers the time of checking the validity of signatures, the correctness of proofs, and the non-conflicts of messages. The results show that our verification algorithm only takes approximately $200$ milliseconds, which is remarkably efficient. The broadcasting time indicates the length of time in propagating a message from one node to another. We assume that all the nodes share the same timestamp, where such configurations can be achieved by NTP service~\cite{mills1991internet}. When the message is generated by a node, it is broadcast to peers in the network. We adopt the concept of \textit{coverage rate} to represent the percentage of reached nodes. Our experiment results show that it takes around $500$ milliseconds to achieve $100\%$ coverage rate for a message, which is the main bottleneck of our system.


\begin{figure}[htb!]
\centering
\begin{tikzpicture}
\begin{axis}[
    xbar stacked,
    legend style={
    legend columns=4,
        at={(xticklabel cs:0.5)},
        anchor=north,
        draw=none
    },
    ytick=data,
    axis y line*=none,
    axis x line*=bottom,
    tick label style={font=\footnotesize},
    legend style={font=\footnotesize},
    label style={font=\footnotesize},
    xtick={0,100,200,300,400,500},
    width=.48\textwidth,
    bar width=3mm,
    xlabel={Time in ms},
    yticklabels={Avg,Max,Min},
    xmin=0,
    xmax=600,
    area legend,
    y=6mm,
    enlarge y limits={abs=0.625},
]
\addplot[findOptimalPartition,fill=findOptimalPartition] coordinates
{(13,0) (30,1) (9,2) };

\addplot[storeClusterComponent,fill=storeClusterComponent] coordinates
{(39,0) (55,1) (38,2) };

\addplot[dbscan,fill=dbscan] coordinates
{(276,0) (318,1) (232,2) };

\addplot[constructCluster,fill=constructCluster] coordinates
{(5,0) (9,1) (3,2) };

\legend{generation,verification, propagation, response}
\end{axis}  
\end{tikzpicture}
\caption{Execution Time (ms) of Different Operations}
\label{fig:bcr}
\end{figure}
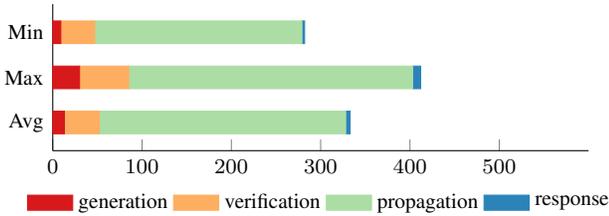

Furthermore, we check the response time of a message. The response time represents the length of time in replying a message from any chain node to the client. This operation averagely takes approximately 15 milliseconds, which is fast and efficient in our implementation. 

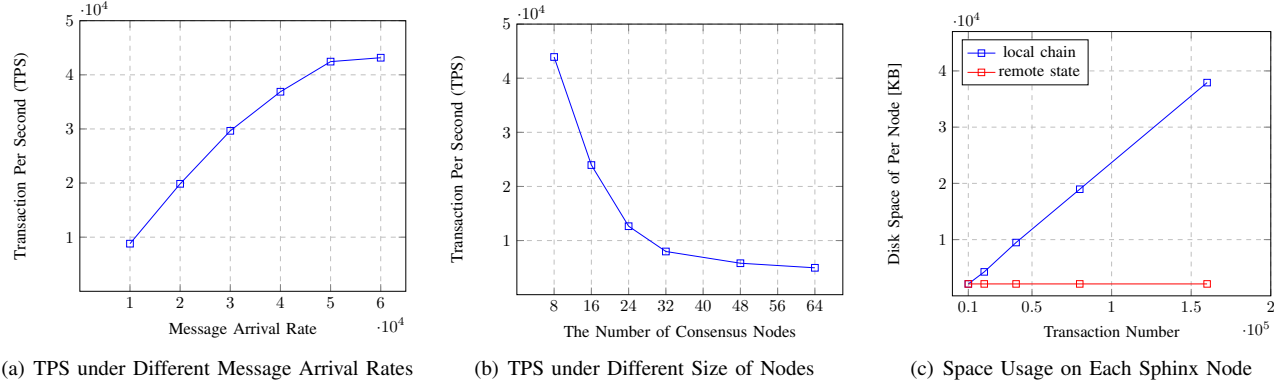
\begin{figure*}[htb!]
\caption{Scalability and Disk Space Evaluation}
\label{tab-scalab}
\centering

\subfigure[TPS under Different Message Arrival Rates]{
\resizebox{0.3\textwidth}{!}{
\begin{tikzpicture}
\begin{axis}[
    xlabel={Message Arrival Rate},
    ylabel={Transaction Per Second (TPS)},
    xmin=0, xmax=65000,
    ymin=0, ymax=50000,
    xtick={10000,20000,30000,40000,50000,60000},
    ytick={10000,20000,30000,40000,50000,60000},
    legend pos=north west,
    ymajorgrids=true,
    xmajorgrids=true,
    grid style=dashed,
]

\addplot[
    color=blue,
    mark=square,
    ]
    coordinates {
    (10000,8800)(20000,19860)(30000,29654)(40000,36874)(50000,42440)(60000,43160)
    };
\end{axis}
\end{tikzpicture}
    \label{fig:model}
}
}
\subfigure[TPS under Different Size of Nodes]{
\resizebox{0.3\textwidth}{!}{
\begin{tikzpicture}
\begin{axis}[
    xlabel={The Number of Consensus Nodes},
    ylabel={Transaction Per Second (TPS)},
    xmin=0, xmax=70,
    ymin=0, ymax=50000,
    xtick={8,16,24,32,40,48,56,64},
    ytick={10000,20000,30000,40000,50000},
    legend pos=north west,
    ymajorgrids=true,
    xmajorgrids=true,
    grid style=dashed,
]

\addplot[
    color=blue,
    mark=square,
    ]
    coordinates {
    (8,43920)(16,23954)(24,12654)(32,7981)(48,5813)(64,4966)
    };
\end{axis}
\end{tikzpicture}
    \label{fig:model}
}
}
\subfigure[Space Usage on Each Sphinx Node]{
\resizebox{0.3\textwidth}{!}{
\begin{tikzpicture}
\begin{axis}[
    xlabel={Transaction Number},
    ylabel={Disk Space of Per Node [KB]},
    xmin=0, xmax=200000,
    ymin=0, ymax=47000,
    xtick={10000,50000,100000,150000,200000},
    ytick={10000,20000,30000,40000},
    legend pos=north west,
    ymajorgrids=true,
    xmajorgrids=true,
    grid style=dashed,
]

\addplot[
    color=blue,
    mark=square,
    ]
    coordinates {
    (10000,2120)(20000,4260)(40000,9480)(80000,18960)(160000,37920)
    };
    \addplot[
    color=red,
    mark=square,
    ]
    coordinates {
    (10000,2120)(20000,2120)(40000,2120)(80000,2120)(160000,2120)
    };
    \legend{local chain, remote state}
\end{axis}
\end{tikzpicture}
    \label{fig:model}
}
}
\end{figure*}

\smallskip
\noindent\textbf{Scalability.}
The scalability in Sphinx is used to describe the capability to handle an increasing number of transactions. To study its scalability, we set the block size at 4MB and block generation rate at $3$s. Then, we run experiments with the following configurations: (i) increase the transaction arrival rate from the client in a fixed number of nodes; (ii) increase the number of nodes with a fixed transaction arrival rate. To obtain an accurate testing result, we repeat the experiments $300$ times and calculate their average performance.

The first experiment attempts to evaluate the average TPS over different transaction arrival rates. We (randomly) set $8$ nodes involved in the system and start the testing at a rate of $2000$ tx/s. Then, we increase the rate in a fixed interval, as it is shown in \underline{Fig.\ref{tab-scalab}.(a)}. The evaluation shows that when the arrival rate is less than $50$k, the throughput increases linearly. When the arrival rate gets close to or above the saturation point ($50$k), the throughput flats out at around $43$k TPS. The reason is that the propagation latency becomes the primary bottleneck, which makes the influence of the arrival rate negligible. Meanwhile, the number of transactions pended in the verification phase also affects the final results.


The second experiment attempts to evaluate the throughput with respect to an increasing number of nodes. Our algorithm, implemented based on PBFT, is only suitable for permissioned blockchain. We limit the size of committee to an upper bound of $64$, and the number of these nodes will remain stable. We send transactions at a fixed rate ($2000$ tx/s) and adjust the participated nodes from $8$ to $64$. \underline{Fig.\ref{tab-scalab}.(b)} shows the throughput of the system drops down along as the participants increase. We go back to the fact that the total throughput of a system is calculated by the multiplication of the number of participated nodes and the TPS of each individual node. Merely increasing the number of participants improves the concurrency, but scarifies the performance of individual nodes. Each of them has to wait for enough replies from peers to exceed the threshold. Thus, the scalability of Sphinx cannot be improved without limitation. 

\smallskip
\noindent\textbf{Disk Space.}
To check the feasibility, we also provide the evaluations of disk space. In Sphinx, the storage of the node grows each time the messages appended into the local chain. Meanwhile, for checking the correctness of the behaviors from other nodes, each auditor must store the state received from other nodes. Thus, we consider space evaluation in two aspects: (i) the size of the local chain, and (ii) the size of the remote state. We assume that there are eight nodes with the transaction creation rate of $100$ messages/second. Then, we monitor the space usage of each node and analyze their growth rates. As shown in \underline{Fig.\ref{tab-scalab}.(c)}, the results indicate that the size of the local chain grows linearly with the increased transactions. On average, the size of the local chain in each node grows at the rate of $0.212$ KB per message. In contrast, the size of the remote state is static which is independent of the scale of transactions. This is easy to understand because the disk usage of remote states relies on the scale of the participated nodes, where the number is fixed in the initial configuration.

\section{Use Cases of Our Consensus}
\label{sec-example}


\noindent\textbf{Certificate System.} Issuing and verifying certificates are slow and complicated, since errors and fraud threaten their usability. A blockchain-based certificate system usually uploads the certificate metadata to the blockchain to achieve reliable management. However, current systems such as Bitcoin suffer from extremely low performance, making the certificate confirmed slowly. This greatly limits the wide adoption of classic blockchain systems. We observed that the key idea behind a certificate system is to store certificates transparently rather than to sort their orders. Thus, our system can perfectly meet the requirements: (i) the transaction data on the chain proves the existence of uploaded certificates; (ii) the high-performance system without linear ordered sequence makes it applicable to the large-scale certificates scenarios.

\smallskip
\noindent\textbf{Log System.} Blockchain technology provides a new approach for the log system since it provides a publicly accessible bulletin board. Recording the logs that are generated by the software onto a distributed storage system greatly improves security. The irreversibility of the blockchain system guarantees the uploaded logs cannot be easily falsified. However, the low performance of current blockchain systems significantly retards the procedures (upload/store/download/change) of logs. This impedes their applications to business scenarios. Our weak consensus benefits current approaches with the ability to prove the existence as well as high performance of processing, enabling blockchain-based log systems practical.

\section{Related Work}
\label{sec-relatedwork}

Our scheme adopts the model of \textit{parallel chains} \cite{fitzi2018parallel}, where each node maintains their own chains. Generally, two types of consensus mechanisms are adopted in our model: the variants of BFT protocols called \textit{leaderless BFT protocols}, and the modified NC protocols named \textit{extended Nakamoto consensus}. 

\smallskip
\noindent\textbf{Leaderless BFT protocols.} \textit{Hashgraph} \cite{baird2016swirlds} was proposed with the leaderless BFT  mechanism. Each node maintains a separate chain, but they are required to interact via the gossip protocol mutually. The node that receives the synchronization information, creates a message locally to record the history, and then broadcast it to peers. Other nodes iterate the same procedure. Hashgraph achieves the consensus through an asynchronous Byzantine consensus. However, the disseminated information containing all previous histories is heavy, which significantly increases communication overhead. Parallel chains in \textit{DEXON} \cite{chen2018dexon} confirm each other through the reference field called $\mathsf{ack}$, and achieve consensus through these references. The consensus consists of three steps. Firstly, each block is deterministically arranged into only one single chain by comparing the residue of its hash value. Then, DEXON employs the technique of a variant BFT protocol stemmed from Algorand \cite{gilad2017algorand} as its single-chain consensus. Finally, it proposes a sorting mechanism to determine the total order of all blocks across parallel chains. In the case of a network delay, these steps will easily be congested, which will negatively affect the entire system. \textit{Aleph} \cite{gkagol2019aleph} is a leaderless BFT distributed system. Each node concurrently issues units (carrying messages). The units are organized in different sets. Units in these sets undergo a voting algorithm. The unit is considered as valid if it receives more votes than the threshold. However, this procedure still relies on linearization, slowing down the performance if any conflict exists. 

\smallskip
\noindent\textbf{Extended Nakamoto Consensus.}  \textit{OHIE} \cite{yu2018ohie} shares similarities in composing multiple parallel chains. OHIE adopts classic Nakamoto consensus for each individual chain. Miners need to calculate a puzzle to generate blocks and additionally have to sort all the blocks. Blocks arrive at a global order across all parallel chains, hence achieving consistency. However, the totally ordered sequence limits the upper bound of performance due to the strong assumption of consistency. \textit{Prism} \cite{bagaria2019prism} structures the network with three types of blocks: proposer blocks, transaction blocks, and voter blocks. These blocks replace the functionalities of a common block in Nakamoto consensus. The consistency is achieved by sorting all transaction blocks. The total ordering is ensured by its proposer blocks, which are selected by voter blocks. However, even decoupled functionalities reach their upper bounds, the procedure of a total ordering algorithm still becomes the bottleneck of throughput. \textit{Chainweb} \cite{will2019chainweb} aims to scale Nakamoto consensus by maintaining multiple parallel chains. It is based on a PoW consensus that incorporates each other’s Merkle roots to increase the hash rate. Each chain in the network mines the same cryptocurrency, which can be transferred cross-chain via a simple payment verification. The method reduces several coins from one chain and creates equal amounts on another chain. However, Kiffer \textit{et al.} \cite{kiffer2018better}\cite{fitzi2018parallel} argued that Chainweb utilizing Nakamoto consensus, is bounded by the same throughput under a same consistency guarantee.

\section{Conclusion}
\label{sec-conclusion}

In this paper, we propose a weak consensus algorithm by relaxing the strong consistency promise. We apply our algorithm to a high-performance blockchain system, \textit{Sphinx}. The system runs with parallel chains, where all transactions and blocks are concurrently processed. We further define the security of \textit{relative persistence} and \textit{liveness} and prove that our system achieves these properties. Also, we provide a full implementation with all layered components, including P2P/consensus/ledger/etc. The evaluation indicates that Sphinx achieves high performance with approximately 43k TPS in total. We further explore the potential applications to demonstrate feasibility and applicability.

\smallskip
\noindent\textbf{Acknowledgement.}
Rujia Li was supported by the National Science Foundation of China under Grant No. 61672015 and Guangdong Provincial Key Laboratory (Grant No. 2020B121201001). Also, the authors would like to thank their supervisors: Qi Wang (Southern University of Science and Technology), David Galindo (University of Birmingham), Yang Xiang (Swinburne University of Technology), Shiping Chen (CSIRO Data61), and team members from HPB foundation for their constructive suggestions on this work.

{\footnotesize \bibliographystyle{IEEEtran}
\bibliography{bib}}

\end{document}